\documentclass[aps,pra,reprint,floatfix,superscriptaddress]{revtex4-1}
\pdfoutput=1
\usepackage[ascii]{inputenc}
\usepackage[bookmarksnumbered,hypertexnames=false]{hyperref}
\usepackage{bbm,microtype,mathrsfs,amsmath,amssymb,xcolor,amsthm,graphicx,cleveref,Quantum,times,tikz,bm}
\usepackage[T1]{fontenc}
\usepackage[USenglish]{babel}
\usetikzlibrary{calc,shapes,positioning,arrows,intersections,decorations.markings,decorations.pathreplacing}
\tikzset{>=latex}
\tikzset{vertex/.style={circle, minimum size=2pt, fill=black}}
\tikzset{->-/.style={thick,decoration={markings, mark=at position 0.9 with {\arrow{>}}},postaction={decorate} } }
\tikzset{->n/.style={thick,decoration={markings, mark=at position #1 with {\arrow{>}}},postaction={decorate} } }
\tikzset{gate/.style={thick,decoration={markings, mark=at position 0.4 with { \node at (0,0) [rectangle,draw=black,fill=white,inner sep=2pt, minimum size=0]{\tiny{#1}}; }}, postaction={decorate} } }
\tikzset{col/.style={thick,decoration={markings, mark=at position 0.4 with { \node at (0,0) [circle,fill=#1,inner sep=2pt, minimum size=0]{}; }}, postaction={decorate} } }
\tikzset{col2/.style 2 args={thick,decoration={markings,
      mark=at position 0.35 with { \node at (0,0) [circle,fill=#1,inner sep=2pt, minimum size=0]{}; },
      mark=at position 0.65 with { \node at (0,0) [circle,fill=#2,inner sep=2pt, minimum size=0]{}; }}, postaction={decorate} } }
\newtheorem{thm}{Theorem}\crefname{thm}{theorem}{theorems}
\crefname{lem}{lemma}{lemmas}
\crefname{cor}{corollary}{corollaries}
\allowdisplaybreaks[4]

\newcommand{\beq}{\begin{equation}}
\newcommand{\eeq}{\end{equation}}

\newcommand{\TT}{\mathbb T}
\renewcommand{\SS}{\mathbb S}
\newcommand{\calH}{\mathcal H}
\newcommand{\calV}{\mathcal V}
\newcommand{\calC}{\mathcal C}
\newcommand{\calM}{\mathcal M}
\newcommand{\calN}{\mathcal N}
\newcommand{\calL}{\mathcal L}
\newcommand{\calP}{\mathcal P}
\newcommand{\calT}{\mathcal T}
\newcommand{\ot}{\otimes}
\newcommand{\CADD}{\operatorname{C}_\text{ADD}}
\newcommand{\ssection}[1]{\smallskip\phantomsection\addcontentsline{toc}{section}{#1}\textit{#1.---}}

\begin{document}
\title{Entanglement from Topology in Chern-Simons Theory}
\author{Grant Salton}
\affiliation{Stanford Institute for Theoretical Physics, Stanford University, Stanford CA 94305}
\author{Brian Swingle}
\affiliation{Stanford Institute for Theoretical Physics, Stanford University, Stanford CA 94305}
\affiliation{Department of Physics, Harvard University, Cambridge MA 02138}
\affiliation{Martin Fisher School of Physics, Brandeis University, Waltham MA 02453}
\author{Michael Walter}
\affiliation{Stanford Institute for Theoretical Physics, Stanford University, Stanford CA 94305}
\begin{abstract}
The way in which geometry encodes entanglement is a topic of much recent interest in quantum many-body physics and the AdS/CFT duality.
This relation is particularly pronounced
in the case of topological quantum field theories, where topology alone determines the quantum states of the theory.
In this work, we study the set of quantum states that can be prepared by the Euclidean path integral in three-dimensional Chern-Simons theory.
Specifically, we consider arbitrary 3-manifolds with a fixed number of torus boundaries in both abelian $U(1)$ and non-abelian $SO(3)$ Chern-Simons theory.
For the abelian theory, we find that the states that can be prepared coincide precisely with the set of stabilizer states from quantum information theory.
This constrains the multipartite entanglement present in this theory, but it also reveals that stabilizer states can be described by topology.
In particular, we find an explicit expression for the entanglement entropy of a many-torus subsystem using only a single replica, as well as a concrete formula for the number of GHZ states that can be distilled from a tripartite state prepared through path integration.
For the nonabelian theory, we find a notion of ``state universality'', namely that any state can be prepared to an arbitrarily good approximation.
The manifolds we consider can also be viewed as toy models of multi-boundary wormholes in AdS/CFT.
\end{abstract}
\maketitle

Quantum information science concerns itself with the transmission and processing of quantum information, tasks that may be described as physical operations on quantum states.
Ultimately each such operation takes the form of unitary time-evolution generated a Hamiltonian, possibly depending on time, which describes the dynamics of all relevant degrees of freedom.
By contrast, much of quantum mechanics and quantum field theory involve the use of \emph{imaginary} or \emph{Euclidean} time evolution to define states of interest.
Imaginary time evolution describes evolution with a Hamiltonian $H$ for an imaginary time $t = - i \tau$, essentially evolution with the non-unitary operator $e^{- \tau H}$.
So far as we know, imaginary time evolution is merely a useful device for calculating properties of certain quantum states, e.g.\ a ground state or a thermal state of $H$.
All states realized in nature must ultimately be prepared from some initial state via real time evolution.
Nevertheless, given the ubiquity and importance of imaginary time evolution, and the more general notion of a \emph{Euclidean path integral}, it is desirable to better understand its expressive power.

Our interest in this work is the way Euclidean path integrals map geometries to states.
To more precisely define what we mean by a Euclidean path integral, we focus on the case of quantum field theories which can be coupled to arbitrary ``background geometries".
For a $D$-dimensional quantum field theory, this means there is a mapping $Z$, the path integral or functional integral, which takes a $D$-dimensional manifold $\calM$ (a spacetime) and returns a probability amplitude $Z(\calM)$.
When the manifold $\calM$ has no boundary, this amplitude is simply a complex number.
When the manifold $\calM$ has a boundary, then the boundary field configuration must be specified and $Z$ gives the amplitude for this field configuration.
In other words, the path integral specifies a state $\ket{\calM}$ in a Hilbert space $\HS_{\partial\calM}$ associated with boundary field configurations.

The simplest Euclidean path integral is defined by taking a space $\Sigma$ and an interval of length $\beta$ and considering the spacetime $\Sigma\times[0,\beta]$.
The corresponding path integral is nothing but the thermal density matrix, $e^{-\beta H_\Sigma}$, where $H_\Sigma$ is the Hamiltonian of the field theory on space $\Sigma$ (which need not be flat space).
If the boundary consists of multiple connected components, then the total Hilbert space is a tensor product $\HS_{\partial\calM}=\ot_{i=1}^n\HS_{\Sigma_i}$, where $\partial\calM=\cup_{i=1}^n\Sigma_i$, and the total Hamiltonian is $H_{\partial\calM}=\sum_{i=1}^n H_{\Sigma_i}$.
This means that the different components $\Sigma_i$ are not coupled and so the Euclidean path integral on $\partial\calM\times[0,\beta]$ factorizes into a tensor product of the thermal density matrices for each component.

In general, $\calM$ does not have the form of $\text{space} \times \text{time}$ and the meaning of the path integral is more complicated. In particular, entangled states in $\HS_{\partial\calM}$ can be prepared.
This is perhaps most distinctive in topological quantum field theories~\cite{witten1988topological}, where the short-range dynamics is trivial but the system still yields highly nontrivial states.
For the remainder of this paper, we imagine $\partial \calM$ is fixed.
Then every $D$-dimensional manifold $\calM$ with the right boundary defines a state in the Hilbert space $\HS_{\partial\calM}$.
The question we pose is the following:
\emph{What is the set of states in $\HS_{\partial\calM}$ that arise by varying $\calM$?}
This question makes sense for any quantum field theory, but in this work we will focus on topological quantum field theories~\footnote{One usually says that any state of a field theory can be prepared by Euclidean path integral provided one can also insert arbitrary operators into the path integral. In our analysis, we do not allow ourselves to insert arbitrary operators as we are specifically interested in the way the path integral maps topology to entanglement.}.
These theories provide the simplest setting in which to investigate the Euclidean path integral as a mapping between geometry and states, and we will be able to give a complete answer to the main question.

Informally, our results are as follows.
We study three-di\-men\-si\-o\-nal \emph{Chern-Simons theory}~\cite{witten89jones}, a topological quantum field theory which is closely related to the physics of the fractional quantum Hall effect~\cite{zhang1989effective}.
For simplicity, we specialize to the case where the boundary of $\calM$ is a union of $n$ tori, $\Sigma_i = \TT^2$.
The definition of Chern-Simons theory, which we recall below, depends on a choice of gauge group $G$ and of a level $k$.
We first study an abelian theory, $G = U(1)$, and show that the class of quantum states that can be prepared by the path integral coincides with the class of \emph{stabilizer states}~\cite{gottesman96stabilizer}, which play an important role in quantum information theory.
This result has two important consequences:
First, it allows us to derive a purely topological formula for the tripartite entanglement.
Previously, only bipartite entanglement entropies had been evaluated.
Second, it shows that stabilizer states can be naturally parameterized by topology, which is quite different from their well-known description in terms of stabilizer generators or discrete phase space.
We then study a non-abelian theory, $G = SO(3)$, for which we show that, by contrast, state preparation by the path integral for certain levels $k$ is \emph{universal}.
Specifically, we can prepare an arbitrarily good approximation to any state in the boundary Hilbert space by a suitable choice of manifold $\calM$.

Besides the intrinsic scientific interest in better understanding the physics of the Euclidean path integral, another motivation for our investigation comes from AdS/CFT duality, and in particular the ER=EPR conjecture~\cite{swingle2012entanglement,van2010building,maldacena2013cool}.
In that context there are so-called multi-boundary wormhole solutions~\cite{brill2000black,skenderis2011holography} which describe bulk gravitational geometries that connect disconnected boundary spacetime regions $\Sigma_i \times \text{time}$.
These bulk geometries describe states of otherwise decoupled conformal field theories living on $\Sigma_i \times \text{time}$, states whose multipartite entanglement properties have been of recent interest~\cite{hayden2013holographic,balasubramanian2014multiboundary,susskind2014erepr,bao2015holographic,marolf2015hot,susskind2016copenhagen,nezami2016multipartite}.
In some cases the relevant states can be understood as arising from a conformal field theory Euclidean path integral evaluated on a bulk spatial splice.
Some of the relevant conformal field theories involve sub-sectors that include the topological interactions we study, so our work can also be regarded as a small step towards a microscopic understanding of these multi-boundary wormholes. These motivations imagine the Chern-Simons theory as a toy model of the boundary field theory in AdS/CFT duality, but there are also connections between bulk gravity and Chern-Simons theory (typically of non-compact groups) that would be interesting to explore further from the perspective of our results~\cite{witten1988twoplusone,gukov2004chern,witten2007three}.

Our work can also be placed in the context of topological quantum computation~\cite{freedman2003topological,nayak2008nonabelian}, which is similarly rooted in topological quantum field theory.
We caution that the notion of universality used in quantum computation is a priori unrelated to the ``state universality'' discussed in this work.
The former concerns the capabilities of a computational model, while the latter concerns the \emph{physical} state space of the topological theory.
However, it has been observed that some schemes that are not universal for topological quantum computation using quasiparticle braiding alone can be made universal by incorporating certain topology-changing operations~\cite{freedman2006towards,freedman2005tilted,bonderson2013twisted,barkeshli2013twist,bonderson2016twisted,barkeshli2016modular}. Indeed, both notions of universality are closely related to properties of the mapping class group representation induced by the topological theory, and our universality result for the nonabelian theory relies on the density theorem of~\cite{larsen2005density}.

\ssection{Chern-Simons Theory}%
Let $G$ be a compact connected Lie group, the gauge group.
In three dimensions, the \emph{Chern-Simons action} for a closed, connected, oriented manifold $\calM$ is defined as
\begin{equation}\label{eq:cs action}
  S_\text{CS} = \frac k {4\pi} \int_\calM \tr\Bigl(A\wedge dA + \frac23 A\wedge A\wedge A\Bigr),
\end{equation}
where $A$ is a Lie algebra-valued connection 1-form, the gauge field;
the suitably normalized trace is taken in a defining representation of the Lie algebra of $G$,
and $k$ is a parameter~\footnote{Strictly speaking, we must treat the gauge field as a connection and sum over non-trivial bundles when the group $G$ is not simply connected. For example, the partition function for closed three-manifolds can be defined by extending the three-manifold to a four-manifold whose boundary is the three-manifold. The Chern-Simons action is then defined by integrating $\tilde{F} \wedge \tilde{F}$ over the four-manifold where $\tilde{F} = d\tilde{A} + \tilde{A} \wedge \tilde{A}$ and $\tilde{A}$ is an extension of $A$ to the four-manifold (e.g., \cite{witten89jones,guadagnini2014path}).
Alternatively, surgery~\eqref{eq:surgery} can be used to define Chern-Simons theory for general three-manifolds $\calM$ in terms of $\SS^3$
~\cite{reshetikhin1991invariants,turaev94quantum,guadagnini1993link}.}.
Under a gauge transformation $g\colon\calM\to G$, the gauge field transforms as $A\mapsto g^{-1}A g+g^{-1}dg$.
A natural class of gauge-invariant observables is given by \emph{Wilson loops}
\[ W(C,R) = \tr_R \calP \exp\Bigl( i \oint_C A \Bigr), \]
defined as the trace of the holonomy of the gauge field around an oriented closed curve $C$ in an irreducible representation $R$ of $G$; $\calP$ denotes path ordering along $C$.
More generally we define for any oriented link $L=C_1\cup\dots\cup C_n$ and representations $R_1,\dots,R_n$ a corresponding observable $W(L,R) = W(C_1,R_1) \cdots W(C_n,R_n)$.
When the group $G$ is nonabelian then the action $S_\text{CS}$ itself is \emph{not} invariant under gauge transformation, but instead changes by an amount proportional to $2\pi k$.
In order to obtain a gauge-invariant quantum field theory by the Feynman path integral, following Witten~\cite{witten89jones}, we need to require that $\exp(i S_\text{CS})$ be invariant; this constrains $k$ to be an integer.
We call $k$ the \emph{level} of the Chern-Simons theory with gauge group $G$~\footnote{Sometimes the level refers to the fully renormalized coupling constant (see~\cite{guadagnini1993link} for a comprehensive discussion).}.

Even on the heuristic level of the Feynman path integral, additional topological structure is required so that the theory is well-defined.
In particular, the expectation values of Wilson loop operators are only well-defined if we consider so-called \emph{framed links}, where the closed curves are thickened to solid tori, and modify the definition of $W(L,R)$ accordingly~\footnote{The manifold also needs to be endowed with a 2-framing. The standard choice is Atiyah's canonical 2-framing~\cite{atiyah1990framings}.}.
In this way, the expectation value
\begin{align}\label{eq:expectation value}
	\langle W(L,R) \rangle_{\calM} &= \frac {Z(\calM; L, R)} {Z(\calM)}
\end{align}
where
\[
  Z(\calM; L, R) =\! \int \mathcal DA \, e^{iS_\text{CS}} W(L,R), \;
  Z(\calM) =\! \int \mathcal DA \, e^{iS_\text{CS}}
\]
becomes a topological invariant of oriented framed links. 
This expression is purely formal due to the 
Feynman path integral; mathematically rigorous definitions 
have been proposed e.g.\ in~\cite{reshetikhin1991invariants,turaev94quantum,blanchet1995topological,manoliu1998abelian,manoliu1998abelian2}.
Note that for a fixed manifold $\calM$, $\langle W(L,R) \rangle_{\calM}$ is proportional to $Z(\calM; L, R)$; this will be a source of many proportionality signs in the following.

\begin{figure}
\centering
\includegraphics[width=0.6\columnwidth]{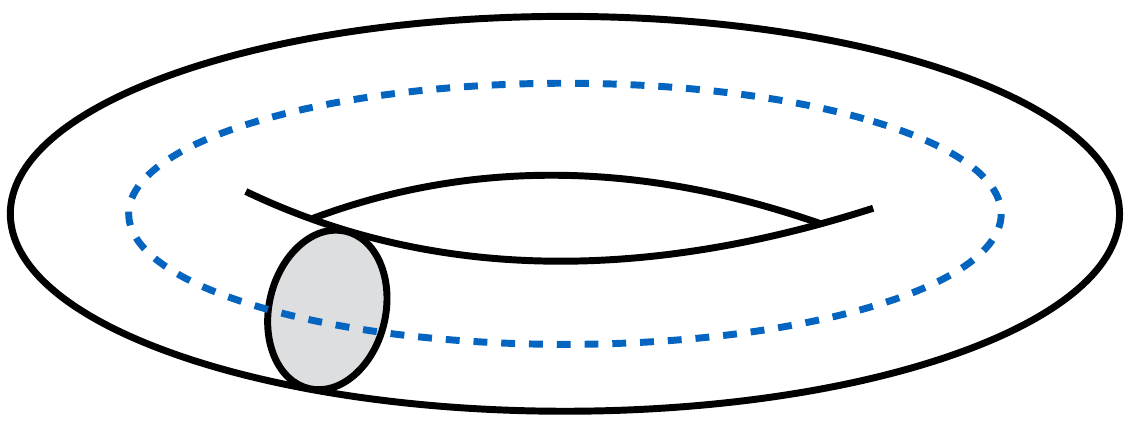}
\caption{\label{fig:Rjtorus}A solid torus with a Wilson loop inserted into its core.
Through the path integral, this object determines a quantum state.
The states $\ket j$ obtain from certain irreducible representations $R_j$ form an orthonormal basis for $\HS_{\TT^2}$.}
\end{figure}

When $\calM$ is a manifold with boundary then the path integral, possibly with Wilson loop operators inserted, should determine a quantum state $\ket\calM$ in a Hilbert space $\HS_{\partial\calM}$ spanned by inequivalent gauge field configurations on the boundary.
A distinguished role is played by the solid torus $\calT$, with boundary the two-dimensional standard torus $\TT^2$. 
It supports a nontrivial Wilson loop operator $W(C;R)$ along its non-contractible core, and the path integral with this operator inserted prepares a quantum state in the boundary Hilbert space (\cref{fig:Rjtorus}).
Not all states obtained in this way are distinct; instead, for each level $k$ there is a finite set of irreducible representations $R_j$ such that the corresponding states $\ket j$ form an orthonormal basis of $\calH_{\TT^2}$.
We note that changing the orientation amounts to passing to the dual Hilbert space, $\calH_{-\Sigma}\cong\calH_\Sigma^*$.

Now suppose that $\calM$ is a manifold with $n$ torus boundaries $\Sigma_j=\TT^2$. 
Let $\overline\calM$ denote the closed manifold obtained by gluing solid tori, with cores $C_1,\dots,C_n$.
Then the multipartite quantum state $\ket{\calM}$ in $\calH_{\TT^2}^{\ot n}$ prepared by the path integral can be computed as follows:
\begin{equation}\label{eq:M state}
	\braket{j_1,\dots,j_n}\calM = \langle W(C_1,R_{j_1}^*) \dots W(C_n,R_{j_n}^*)\rangle_{\overline\calM}
\end{equation}
If we are in addition given a framed link $L$ in $\calM$ labeled by representations $R$ then we denote the corresponding boundary state by $\ket{\calM;L,R}$; its amplitudes can be computed by inserting a further Wilson loop operator $W(L,R)$ into~\eqref{eq:M state}.

\smallskip

We are now in a position to restate our overarching question more precisely:
Fixing a gauge group $G$, a level $k$ and a number $n$, we are interested in the states $\ket{\calM;L,R}$ that can be obtained when we vary over all (not necessarily connected) 3-manifolds $\calM$ with $n$ torus boundaries, a well as over framed links $L$ in $\calM$ whose components are labeled by irreducible representations $R$.
Note that, in general, we will compute unnormalized states without worrying about overall normalizations.
We focus on torus boundaries for simplicity; our results extend easily to higher genus boundaries. With this choice our Hilbert space will be $\HS_{\TT^2}^{\ox n}$ -- a tensor product of $n$ Hilbert spaces, each of which is spanned by distinguished basis vectors $\ket j$, which will play the role of our computational basis.
It is instructive to observe that any state $\ket{\calM;L,R}$ can be obtained from a manifold without Wilson loops by projecting on a computational basis state (since we restrict to irreducible representations in $R$).

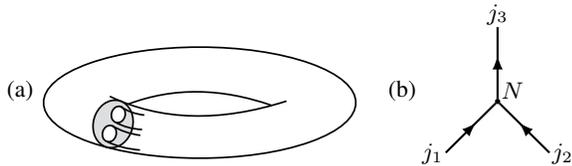
\begin{figure}[!t]
	\raisebox{.9cm}{(a)}\includegraphics[width=0.5\columnwidth,trim={220 0 0 0},clip]{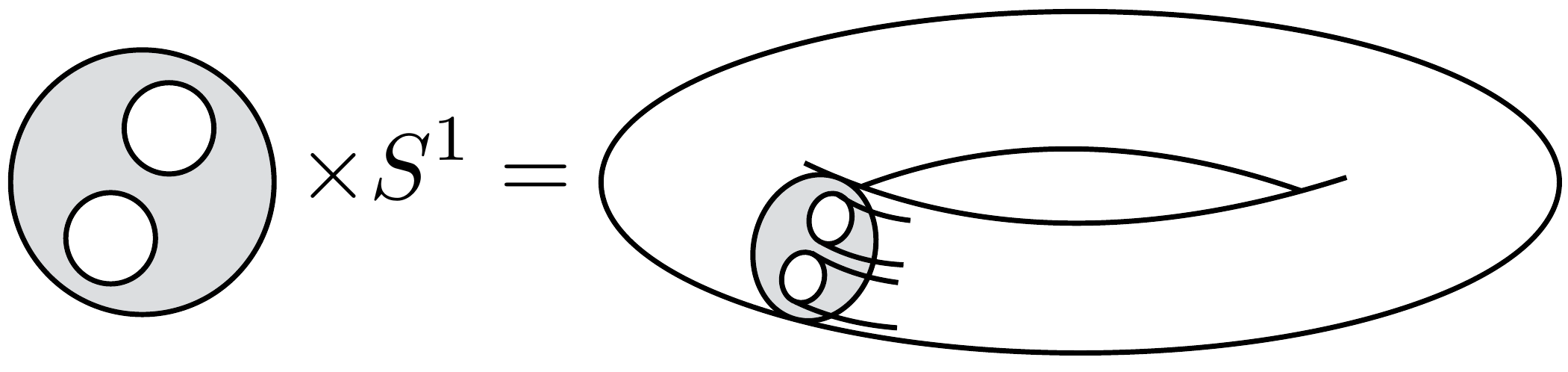}
	\quad
	\raisebox{.9cm}{(b)}
\begin{tikzpicture}[inner sep=0pt, outer sep=0pt]
  \node [label=left:$j_1$](leg1) at (-0.707, 0){};
  \node [label=right:$j_2$](leg2) at (0.707, 0){};
  \node [label=above:$j_3$](leg3) at (0,1.707){};
  \node[vertex, label={[xshift=0.2cm]$N$}](vertex1) at (0,0.7){};
  \draw[->n=0.6] (leg1) -- (vertex1);
  \draw[->n=0.6] (leg2) -- (vertex1);
  \draw[->n=0.6] (vertex1) -- (leg3);
\end{tikzpicture}
	\caption{\label{fig:fusetensor}
	(a) A disk with two punctures rotated around an $\SS^1$ forms a solid torus with two tori removed from the interior. This manifold $\calN$ has three toroidal boundaries.
	(b) Our notation for the corresponding fusion tensor $N^{j_3}_{j_1j_2}$ obtained by the path integral over $\calN$.}
\end{figure}

As a first example, let $\calN$ denote the manifold obtained by taking a solid torus and removing two smaller tori from the interior, as shown in \cref{fig:fusetensor},~(a).
This object has three boundary tori, two of which are negatively oriented, so that the path integral determines a tripartite state $\ket\calN$ in $\calH_{\TT^2}^* \ot \calH_{\TT^2}^* \ot \calH_{\TT^2}^*$, best thought of as a map from two copies to one copy of the Hilbert space.
We can compute its components using~\eqref{eq:M state} in terms of expectation values of Wilson loop operators in $\SS^2\times \SS^1$, the manifold obtained from $\calN$ by filling in the three boundary tori with cores $C_1,C_2,C_3$.
The result is
\begin{equation}
\label{eq:fusion}
	\langle W(C_1,R_{j_1}) W(C_2,R_{j_2}) W(C_3,R_{j_3}^*)\rangle_{\SS^2 \times \SS^1} \propto N_{j_1j_2}^{j_3},
\end{equation}
where $N_{j_1j_2}^{j_3}$ denotes the \emph{fusion tensor}, which we represent graphically as a three leg tensor, shown in \cref{fig:fusetensor},~(b).
It encodes the fusion rules of the Chern-Simons theory with gauge group $G$ and level $k$.
The appearance of the fusion rules here is not surprising, since the three Wilson loops wrapping the $\SS^1$ correspond to three punctures on $\SS^2$; but the $2$-sphere can only support the charges when the fusion rules are satisfied.
Since $N_{j_1j_2}^0 = \delta_{j_1,j_2^*}$, the manifold obtained by gluing a solid torus to the outside of $\calN$ will, after an orientation reversal, prepare a maximally entangled state $\sum_j \ket{j,j^*}$.

It is often quite complicated to compute the state~\eqref{eq:M state} directly, and it can be useful to cut the manifold $\calM$ along a closed two-dimensional surface $\Sigma$ to obtain simpler three-manifolds $\calM_1$ and $\calM_2$.
In other words, we construct the manifold $\calM$ by \emph{gluing} $\calM_1$ and $\calM_2$ along $\Sigma$. 
The path integral is compatible with gluing and so $\ket{\calM}$ can be obtained by contracting the state $\ket{\calM_1} \ot \ket{\calM_2} \in \HS_{\partial \calM} \ot \HS_\Sigma \ot \HS_\Sigma^*$ along the $\Sigma$-factors.
We will repeatedly use this strategy in the following.

In addition to Wilson loops and path integrals, we have other tools at our disposal.
Let $\Sigma$ be a two-dimensional surface. Any orientation-preserving diffeomorphism on $\Sigma$ is represented by an operator on the Hilbert space $\calH_\Sigma$ (it can be defined by the path integral).
Since the theory is topological, we thus obtain a (in general projective) representation of the \emph{mapping class group} (MCG) of $\Sigma$, which is defined as the group of isotopy classes of orientation-preserving diffeomorphisms.
In the case of the torus, $\Sigma=\TT^2$, the MCG is isomorphic to the modular group $SL(2,\ZZ)$, and it is generated by the so-called $S$ and $T$ transformations (which exchange the cycles of the torus and introduce Dehn twists, respectively).
For a general Chern-Simons theory, it is known how to compute their representation matrices in terms of representation theory and level~\cite[(2.23)]{marino05chern}.

Any closed connected orientable manifold $\calM$ can be obtained from the 3-sphere $\SS^3$ via Dehn surgery: cutting out a tubular neighborhood of a link $\calL$ and re-gluing solid tori with homeomorphisms applied along each cut.
The link $\calL = \calC_1\cup\dots\cup\calC_m$ is called a \emph{surgery link} and it carries an orientation and framing.
Accordingly, we obtain \emph{surgery rules}, which compute the expectation values~\eqref{eq:expectation value} for the manifold $\calM$ in terms of the 3-sphere $\SS^3$:
\begin{equation}\label{eq:surgery}
	\langle W(L,R) \rangle_{\calM} = \frac {\langle W(L,R) \tilde W(\calL) \rangle_{\SS^3}} {\langle \tilde W(\calL) \rangle_{\SS^3}}
\end{equation}
Here, $\tilde W(\calL) = \tilde W(\calC_1)\cdots\tilde W(\calC_m)$ is defined in terms operators $\tilde W(\calC) := \sum_\ell S_{\ell0} W(C,R_\ell)$.
Surgery will be useful to compute expectation values and thereby the states $\ket\calM$ that we are interested in classifying.
If the denominator in~\eqref{eq:surgery} is zero then we will consider the theory not to be well-defined on the manifold $\calM$.
We refer to~\cite{dijkgraaf1990topological,guadagnini1993link,marino05chern,freed2009remarks} for further detail on Chern-Simons theory.

\ssection{\texorpdfstring{$U(1)$ Theory}{U(1) Theory}}%
We now turn our attention to a simple abelian Chern-Simons theory, $G=U(1)$.
The irreducible representations of $G$ are labeled by integers $q\in\ZZ$.
Our goal in the following will be to connect the $U(1)$ theory to stabilizer states, which are most interesting in odd dimension.
We will thus consider odd level $k$.
In this case the Wilson loop operators $W(C,R_q)$ depend on the charge $q$ modulo $2k$.
However, there is a natural restriction to the even-valued charges which we use to obtain a $k$-dimensional theory \footnote{There is a subtlety here: $U(1)$ Chern-Simons theory is bosonic when the level is even but fermionic when the level is odd. To be precise, when $k$ is odd the charge $q=k$ object is a fermion which is ``local'' insofar as it has trivial braiding with all other charges. Because the theory contains a local fermion, it can only be defined on manifolds that admit a spin structure. For the torus, this requires picking either periodic or anti-periodic boundary conditions for fermions along the two cycles of the torus. Furthermore, modular transformations mix these boundary conditions, so it is important to restrict to a closed sector of the theory. One can also obtain \cref{U1STmatrix} from this perspective.}.
Specifically, we set $q_j=j+kj$, where $j=0,\dots,k-1$, so that the Hilbert space $\calH_{\TT^2}$ of the torus is spanned by basis states that we denote by $\{\ket j\}_{j=0}^{k-1}$.

The modular $S$ and $T$ transformations have the following unitary representations on $\HS_{\TT^2}$:
\begin{align}
\label{U1STmatrix}
S_{j j'}=\frac1{\sqrt k} \omega^{jj'}, \quad
T_{j j'}=\delta_{j,j'} \omega^{j(j+k)/2},
\end{align}
where $\omega=e^{2\pi i/k}$~\cite{wen1990topological}.
The $S$-matrix implements a discrete Fourier transform, i.e., a base change from the computational basis $\ket j$ into the conjugate basis.
The fusion tensor is given \nolinebreak by:
\begin{equation}\label{eq:U1 fusion}
N^{j_3}_{j_1j_2}=\delta_{j_3,j_1+j_2} =
\raisebox{-0.5cm}{\begin{tikzpicture}[scale=0.6, every node/.style={scale=0.6},inner sep=0pt, outer sep=0pt]
  \node [label=left:$j_1$](leg1) at (-0.707, 0){};
  \node [label=right:$j_2$](leg2) at (0.707, 0){};
  \node [label=above:$j_1+j_2$](leg3) at (0,1.707){};
  \node[vertex](vertex1) at (0,0.7){};
  \draw[->n=0.6] (leg1) -- (vertex1);
  \draw[->n=0.6] (leg2) -- (vertex1);
  \draw[->n=0.6] (vertex1) -- (leg3);
\end{tikzpicture}}
\end{equation}
This tensor comes from $U(1)$ charge conservation, and it has two useful interpretations.
When viewed as a map from two Hilbert spaces into one, it can be written as
$S N (S^\dagger\ot S^\dagger) \propto \sum_j \ketbra j {j,j}$.
Thus its adjoint, which is prepared by a path integral over the manifold in \cref{fig:fusetensor} but with opposite orientation, represents a ``copy'' operation in the basis $\{S^\dagger\ket j\}$, as shown in \cref{fig:Cplustensor},~(a).
Viewed as a tripartite quantum state, the same object corresponds to a genuinely tripartite entangled Greenberger-Horne-Zeilinger (GHZ) state~\cite{greenberger1989going}.
This already indicates that we can produce interesting states by path integration.

Now consider the manifold shown in \cref{fig:fusetensor} but with a solid torus carrying a Wilson loop in the representation $1$ inserted in one of the holes.
Path integration on this manifold results in the fusion tensor with one index fixed to the value $1$:
\begin{equation}\label{eq:PauliXtensor}
N^{j_2}_{j_1,1}=\delta_{j_2,j_1+1}=\raisebox{-0.5cm}{\begin{tikzpicture}[scale=0.6, every node/.style={scale=0.6},inner sep=0pt, outer sep=0pt]
  \node [label=left:$j_1$](leg1) at (-0.707, 0){};
  \node [label=right:$1$](leg2) at (0.707, 0){};
  \node [label=above:$j_1+1$](leg3) at (0,1.707){};
  \node[vertex] (vertex1) at (0,0.7){};
  \draw[->n=0.6] (leg1) -- (vertex1);
  \draw[->n=0.6] (leg2) -- (vertex1);
  \draw[->n=0.6] (vertex1) -- (leg3);
\end{tikzpicture}}
\end{equation}
This tensor, together with its conjugate version, allows us to generate the full single qudit Pauli group.
Recall that the \emph{Pauli group} of a $k$-dimensional system is generated by the operators $X$ and $Z$ defined by $X \ket j = \ket{j+1\pmod k}$ and $Z\ket j=\omega^j\ket j$.
As such, the tensor in~\eqref{eq:PauliXtensor} implements the ``shift operator'' $X$, and the phase operator is given by $Z = S X S^\dagger$.
Thus $Z$ is prepared by a path integral over the same manifold, but with boundary tori identified with $\TT^2$ in a different way.
For $n$ systems, the Pauli group is given by tensor products of generators, corresponding to the disjoint union of the corresponding 3-manifolds.
It follows that we can produce any element of the Pauli group for $n$ qudits by path integration.
From these examples, it is clear that we can produce interesting, non-trivial states; our goal now is to see how far we can go and characterize the set of all states we can prepare.

Before diving into our main results, let us briefly review the Clifford group and the related stabilizer states.
The \emph{Clifford group} is defined to be the normalizer of the Pauli group.
The \emph{stabilizer states} are the subset of the full $n$-qudit Hilbert space that arise as eigenstates of a maximal subgroup of the Pauli group.
Equivalently, they are the states that can be prepared using only Clifford gates applied to the fiducial state $\ket{0}^{\ox n}$.
The Gottesman-Knill theorem~\cite{gottesman1999heisenberg} asserts that any quantum computation involving only Clifford group elements applied to stabilizer states can be simulated efficiently on a classical computer.
As such, the Clifford group and the stabilizer states are a restricted class of unitaries and states that lack universal quantum computational power.
Generators of the Clifford group for odd $k$ are given by the Fourier transform $S$, the \emph{phase gate} $P \ket j = \omega^{j(j-1)/2} \ket j$ and the \emph{controlled addition} $\CADD \ket{j,\ell} = \ket{j,\ell+j}$~\cite{gottesman1999fault,farinholt2014ideal}.
We now state our main result for $U(1)$ Chern-Simons theory:

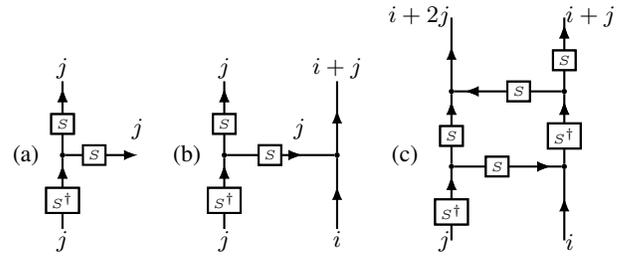
\begin{figure}
\centering
\raisebox{1.2cm}{(a)}
\begin{tikzpicture}[inner sep=0pt, outer sep=0pt]
  \node [label=below:$j$](leg1) at (0, 0) {};
  \node [label=above:$j$](leg3) at (0, 2) {};
  \node [label=above:$j$] at (1,1.2) {};
  \node [vertex](vertex1) at (0, 1) {};

  \draw [->-,gate=$S^\dagger$]   (leg1) -- (vertex1);
  \draw [->-,gate=$S$]  (vertex1) -- (leg3);
  \draw [->,gate=$S$]  (vertex1) -- (1,1);
\end{tikzpicture}
\quad
\raisebox{1.2cm}{(b)}
\begin{tikzpicture}[inner sep=0pt, outer sep=0pt]
  \node [label=below:$j$](leg1) at (0, 0) {};
  \node [label=below:$i$](leg2) at (1.5, 0) {};
  \node [label=above:$j$](leg3) at (0, 2) {};
  \node [label=above:$i+j$](leg4) at (1.5, 2) {};
  \node [label=above:$j$] at (1,1.2) {};
  \node [vertex](vertex1) at (0, 1) {};
  \node [vertex](vertex2) at (1.5, 1) {};

  \draw [->-,gate=$S^\dagger$]   (leg1) -- (vertex1);
  \draw [->-,gate=$S$]  (vertex1) -- (leg3);
  \draw [->n=0.6]     (leg2) -- (vertex2);
  \draw [->n=0.6]     (vertex2) -- (leg4);
  \draw [->n=0.7,gate=$S$]  (vertex1) -- (vertex2);
\end{tikzpicture}
\quad\raisebox{1.2cm}{(c)}\!\!\!\!\!\!\!\!\!
\begin{tikzpicture}[inner sep=0pt, outer sep=0pt]
  \node [label=left:$j$](leg1) at (0, 0) {};
  \node [label=right:$i$](leg2) at (1.5, 0) {};
  \node [label=left:$i+2j$](leg3) at (0, 3) {};
  \node [label=right:$i+j$](leg4) at (1.5, 3) {};
  \node [vertex](vertex1) at (0, 1) {};
  \node [vertex](vertex2) at (1.5, 1) {};
  \node [vertex](vertex3) at (0, 2) {};
  \node [vertex](vertex4) at (1.5, 2) {};

  \draw [->-,gate=$S^{\dagger}$]   (leg1) -- (vertex1);
  \draw [->-,gate=$S$]  (vertex1) -- (vertex3);
  \draw [->n=0.6]     (vertex3) -- (leg3);
  \draw [->n=0.6]     (leg2) -- (vertex2);
  \draw [->-,gate=$S^{\dagger}$] (vertex2) -- (vertex4);
  \draw [->-,gate=$S$]  (vertex4) -- (leg4);
  \draw [->-,gate=$S$]  (vertex1) -- (vertex2);
  \draw [->-,gate=$S$]  (vertex4) -- (vertex3);
\end{tikzpicture}
\caption{\label{fig:Cplustensor}
(a) The copy tensor can be prepared by conjugating the adjoint of the fusion tensor by Fourier transforms.
(b) We then obtain the controlled addition tensor $\CADD$ by contracting this copy tensor with a fusion tensor.
(c) We obtain a perfect tensor from two copies of the $\CADD$ tensor.}
\end{figure}

\begin{thm}\label{U1allstabThm}
For $U(1)$ Chern-Simons theory at odd prime level $k$, we can prepare any stabilizer state in the $n$-torus Hilbert space $\HS_{\TT^2}^{\ot n}=(\CC^k)^{\ot n}$.
Conversely, if $k\equiv1\pmod4$ then we can only prepare stabilizer states.
\end{thm}
\begin{proof}
We will build the remaining Clifford generators $P$ and $\CADD$ explicitly.
Comparing the definition of the phase gate $P$ with~\eqref{U1STmatrix}, we see that $T$ has the same form as the phase gate except for a state-dependent phase $\omega^{j(k+1)/2}$.
We use the Pauli-$Z$ gate that we constructed earlier to acquire this phase.
Indeed, $P = Z^{(k-1)/2} T$, where we note that $(k-1)/2$ is an an integer for odd $k$.
To produce the two-qudit controlled addition $\CADD$, we contract a fusion tensor with a copy tensor, giving rise to a complicated manifold with a simple graphical representation, shown in \cref{fig:Cplustensor},~(b).
It follows from the above that we can prepare any Clifford operator $U$ on $n$ qubits by a path integral over some 3-manifold with $2n$ boundary tori, half of which carry the standard orientation (the outputs) and half of which carry the opposite orientation (the inputs).
If we glue solid standard tori to the latter, we obtain a manifold that produces an arbitrary stabilizer state $U \ket0^{\ot n}$.
This concludes the proof of the first statement.

We now prove the converse statement.
We will use the following important results:
The expectation value of any Wilson loop operator $W(L,R)$ in $\SS^3$ is fully characterized by the linking numbers of its components (the self-linking numbers are determined by the framing).
Specifically, for a framed link $L$ in $\SS^3$ labeled by irreducible representations $R_{j_1},\dots,R_{j_n}$, one has
\begin{align*}
  \langle W(L,R) \rangle_{\SS^3} \propto e^{-\frac{\pi i}k \sum_{a,b} q_{j_a} L_{ab} q_{j_b}} = \omega^{-\sum_{a,b} j_a \tilde L_{ab} j_b},
\end{align*}
where $L_{ab}\in\ZZ$ are the matrix elements of the linking matrix (e.g., \cite{marino05chern}); for the second equality we use that $q_j = j + k j$ and define $\tilde L_{ab} = \frac {1+k^2}2 L_{ab}$, which is again an integer matrix.
According to~\eqref{eq:M state}, the amplitudes of the state $\ket{\SS^3\setminus N}$, where $N$ denotes a tubular neighborhood of the link $L$, are given by the expectation values $\langle W(L,R) \rangle_{\SS^3}$.
On the other hand, the discrete Hudson's theorem~\cite{gross2006hudson} asserts that any state with such amplitudes is a stabilizer state.
Therefore, we conclude that $\ket{\SS^3\setminus N}$ is a stabilizer state.

We now consider an arbitrary manifold $\calM$ with $n$ torus boundaries.
First, we use~\eqref{eq:M state} to express the amplitudes of $\ket{\calM}$ by expectation values of Wilson loop operators in $\overline\calM$, the closed manifold obtained by filling in the boundary tori.
Second, we use the surgery rule~\eqref{eq:surgery} to reduce the latter to an expectation value in $\SS^3$.
The consistency of the surgery rule for $U(1)$ has been proved in~\cite{guadagnini1993link} for prime levels $k\equiv1\pmod4$.
The upshot is that
\begin{align*}
  &\braket{j_1,\dots,j_n|\calM}
\propto \langle W(C_1,R_{j_1}^*) \cdots W(C_n,R_{j_n}^*) \tilde W(\calL)\rangle_{\SS^3}.
\end{align*}
Here, $C_1,\dots,C_n$ denote the cores of the filled-in boundary tori and $\calL$ denotes a surgery link by which we obtain $\calM$ from $\SS^3$;
we have omitted the denominator in~\eqref{eq:surgery}, since it only contributes a global scalar.
Now recall that $\tilde W(\calC)=\sum_j S_{\ell 0} W(\calC, R_\ell)$.
If we denote by $N$ a tubular neighborhood of the link $L\cup\calL$ then it follows from the preceding discussion that
\[ \braket{j_1,\dots,j_n|\calM} \propto \braket{j_1,\dots,j_n,+,\dots,+| \SS^3 \setminus N}, \]
where we have introduced $\ket+ = S\ket 0$.
Thus $\ket{\calM}$ is obtained by contracting the stabilizer state $\ket{\SS^3\setminus N} \in (\CC^k)^{\ot (n+m)}$ with the stabilizer state $\ket+^{\ot m}$, with $m$ the number of components of $\calL$, and therefore again a stabilizer state (e.g., \cite[App.~G]{hayden2016holographic}).
Lastly, we had argued before that if we are given in addition a framed link $L'$ and irreducible representations $R'$ then the corresponding state $\ket{\calM;L',R'}$ can be obtained from a manifold without Wilson loops by projecting on a computational basis state.
By the same argument as before, this shows that $\ket{\calM;L',R'}$ is a stabilizer state.
\end{proof}

\begin{figure}[!t]
\centering
\begin{tikzpicture}[scale=0.6, every node/.style={scale=0.6},inner sep=0pt, outer sep=0pt, baseline=(current bounding box.center)]
  \node[scale=1.67] (label1) at (-2,0){(a)};
  \node[vertex,label=left:$j_1$](leg1) at (-0.707, 0){};
  \node(leg4) at (0, -1.707){};
  \node[vertex,label=right:$j_2$](leg2) at (0.707, 0){};
  \node[label=right:$\;j_1+j_2$](leg3) at (0,1.707){};
  \node[vertex] (vertex1) at (0,0.7){};
  \node[vertex] (vertex2) at (0,-0.7){};
  \draw[->n=0.7] (leg1) -- (vertex1);
  \draw[->n=0.7] (leg2) -- (vertex1);
  \draw[->n=0.8] (vertex2) -- (leg1);
  \draw[->n=0.8] (vertex2) -- (leg2);
  \draw[->n=0.6] (vertex1) -- (leg3);
  \draw[->n=0.7] (leg4) -- (vertex2);
  \draw [thick](0, 1.68) arc[radius=0.6, start angle= 0, end angle= 180];
  \draw [thick] (-1.2,1.68) -- (-1.2,-1.68);
  \draw [thick](0, -1.68) arc[radius=0.6, start angle= 360, end angle= 180];
  \begin{scope}[shift={(4.5,0)}]
  \node[scale=1.67] (label2) at (-2.5,0){(b)};
	\foreach \i in {0,...,2}
	{
 		 \begin{scope}[shift={(\i*2.5,0)}, yscale=(-1)^\i]
                  \node[vertex](vertex1) at (0, -0.5){};
                  \node[vertex] (vertex2) at (0,0.5){};
                  \node (leg1) at (-0.707, -1.207){};
                  \node (leg2) at (0.707, -1.207){};
                  \node (leg3) at (-0.707,1.207){};
                  \node (leg4) at (0.707,1.207){};
                  \draw[->n=0.7] (leg1) -- (vertex1);
                  \draw[->n=0.7] (leg2) -- (vertex1);
                  \draw[->n=0.8] (vertex1) -- (vertex2);
                  \draw[->n=0.8] (vertex2) -- (leg3);
                  \draw[->n=0.8] (vertex2) -- (leg4);
                  \draw [thick](0.68,1.18) arc[radius=0.806, start angle= 135, end angle= 45];
                  \draw [thick](0.68,-1.18) arc[radius=0.806, start angle= 225, end angle= 315];
                  \end{scope}
  	}
  \draw [thick](-0.68,1.18) arc[radius=0.806, start angle= 45, end angle= 155];
  \draw [thick](-0.68,-1.18) arc[radius=0.806, start angle= 315, end angle= 225];
  \draw (6.650,1.14) -- (6.85,1.34);
  \draw (6.592,1.18) -- (6.792,1.38);
  \draw (-1.650,-1.14) -- (-1.85,-1.34);
  \draw (-1.592,-1.18) -- (-1.792,-1.38);

  \draw (6.650,-1.14) -- (6.85,-1.34);
  \draw (-1.80,0.95) -- (-2.0,1.15);

  \node (j1) at (-1.2,1.2) {$j_1$};
  \node (j2) at (1.3,1.2) {$j_2$};
  \node (j3) at (3.8,1.2) {$j_3$};
  \node (j4) at (6.2,1.2) {$j_4$};
  \node at (0.7,0) {$j_1+j_2$};
  \node at (3.2,0) {$j_2+j_3$};
  \node at (5.7,0) {$j_3+j_4$};
  \node at (-1.2,-1.7) {$j_4$};
  \node at (1.3,-1.7) {$j_1+j_2-j_4$};
  \node at (3.8,-1.7) {$j_3+j_4-j_1$};
  \node at (6.2,-1.7) {$j_1$};

\end{scope}
\end{tikzpicture}
\caption{\label{fig:ghz} Computation of (a) $Z(-\calN \cup_{\partial\calN} \calN)=k^2$ and (b) $Z(-3\calN \cup_f 3\calN)=k^4$ for $U(1)$ Chern-Simons theory, where $\calN$ is the manifold from \cref{fig:fusetensor} that induces the fusion tensor~\eqref{eq:U1 fusion}.
Similarly, $Z(-2\calN \cup_{f_A} 2\calN) = k^3$, and hence $S(A)=S(B)=S(C)=\log k$.
As a consequence, we find from~\eqref{eq:tripartite} that the fusion tensor is equivalent to $g = 3 - 2 = 1$ GHZ state.
}
\end{figure}
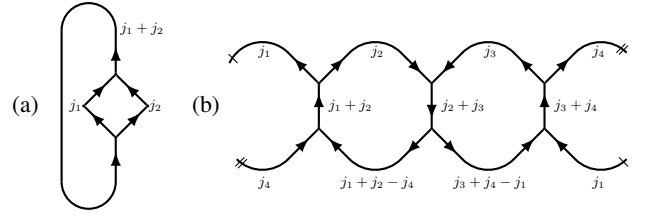

\Cref{U1allstabThm} shows that the path integral of $U(1)$ Chern-Simons theory can be used to prepare arbitrary stabilizer states in the many-torus Hilbert space.
Using the Choi-Jamiolkowski isomorphism, we can also prepare an arbitrary element of the Clifford group.
As a consequence, we can prepare a large class of interesting quantum error correcting codes.
An example is $\sum_{i,j} \ketbra{i+j,i+2j}{i,j}$, which is a perfect tensor, whose construction we show in \cref{fig:Cplustensor},~(c).
Recall that a tensor with an even number of legs is \emph{perfect} if, for any bipartition of the legs, the tensor defines a unitary map from one half of the legs to the other.
Perfect tensors arise in the construction of toy models of holography~\cite{pastawski2015holographic,hayden2016holographic}.

Stabilizer states are in many ways the finite-dimensional analog of Gaussian states~\cite{gross2006hudson}, so it is conceptually quite pleasing that they can be obtained from $U(1)$ Chern-Simons theory, for which the action~\eqref{eq:cs action} is quadratic.
It is also quite intriguing that \cref{U1allstabThm} allows us to parametrize the stabilizer states by topology (though redundantly so), which is rather different from the usual description in terms of stabilizer generators or discrete phase space.

\Cref{U1allstabThm} also has a number of concrete consequences for the entanglement properties of states prepared by the path integral in $U(1)$ Chern-Simons theory.
First, stabilizer states have flat entanglement spectrum.
This means that we can compute the entanglement entropy $S(A)$ of an arbitrary many-torus subsystem $A\subseteq\partial\calM$ by a introducing only a single replica, rather than having to perform an analytic continuation as in the general case (cf.~\cite{dong2008topological,flammia2009topological}).
Explicitly,
\begin{equation}\label{eq:bipartite}
	S(A) 
	= -\log \frac {Z(-2\calM \cup_{f_A}\! 2\calM)} {Z(-\calM \cup_{\partial\calM} \calM)^2},
\end{equation}
where we write $\rho_{\partial\calM} = \ketbra\calM\calM$ for the unnormalized density matrix and $\rho_A = \tr_{\partial\calM\setminus A}(\rho)$ for the reduced density matrix; we denote the path integral over the closed manifold $\calN$ (i.e., \eqref{eq:expectation value} without any operator insertion) by $Z(\calN)$, two copies of the manifold that prepares the pure state by $2\calM = \calM \cup \calM$, and by $f_A$ the diffeomorphism that exchanges the two copies of the tori in $A$, while leaving the tori in $\partial\calM\setminus A$ fixed.
In the presence of Wilson loops, the numerator and denominator of~\eqref{eq:bipartite} include the corresponding Wilson loop operators.

Second, going beyond entanglement entropies, the tripartite entanglement in stabilizer states is well-understood: Any pure tripartite stabilizer state is equivalent to a collection of bipartite Bell pairs and tripartite GHZ states~\cite{bravyi2006ghz,looi2011tripartite}.
Let $\partial\calM = A\cup B\cup C$ denote an arbitrary tripartiton of the boundary tori.
Then the amount of tripartite entanglement, i.e., the number of GHZ states that can be distilled by local unitaries between $A$, $B$ and $C$, is given by the following formula~\cite[(5)]{nezami2016multipartite}:
\begin{equation}\label{eq:tripartite}
	g = \frac {S(A) + S(B) + S(C)} {\log k} + \log_k \frac {Z(-3\calM \cup_f 3\calM)} {Z(-\calM \cup_{\partial\calM} \calM)^3}
\end{equation}
Here, $f$ denotes the diffeomorphism that cyclically exchanges the three copies of the tori in $A$ as $1\mapsto2\mapsto3$, the three copies of the tori in $B$ as $1\mapsto3\mapsto2$, and which leaves the tori in $C$ invariant.
Since the entanglement entropies can be evaluated using~\eqref{eq:bipartite}, \cref{eq:tripartite} is a fully topological formula for evaluating the tripartite entanglement of many-torus states prepared by the path integral in $U(1)$ Chern-Simons theory.
As an example, we can use \cref{eq:tripartite} to verify that the fusion tensor~\eqref{eq:U1 fusion} is indeed equivalent to a GHZ state (see \cref{fig:ghz}).

\ssection{\texorpdfstring{$SO(3)$ Theory}{SO(3) Theory}}%
\begin{figure}
\centering
\raisebox{0.7cm}{(a)} \includegraphics[width=0.75\columnwidth]{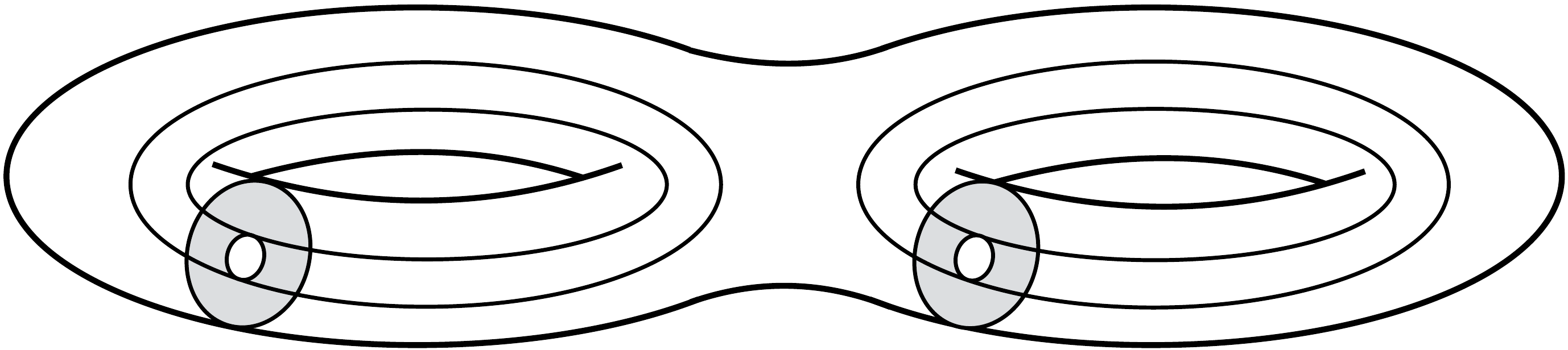} \\[0.3cm]
\raisebox{0.6cm}{(b)} \includegraphics[width=0.85\columnwidth]{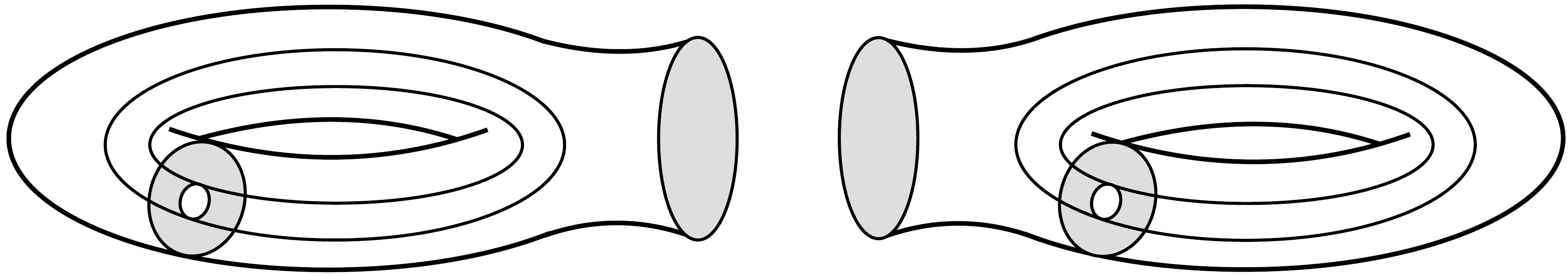}
\caption{\label{fig:2g1tog2}
(a) A genus-2 handlebody with two solid tori removed.
The boundary of this manifold $\calV$ consists of a genus-2 surface $\Sigma_2$ and two tori, and it induces by path integration an isometry $\HS_{\TT^2}^{\ot 2}\to\HS_{\Sigma_2}$.
(b) The same manifold, cut into two pieces along an embedded disk.}
\end{figure}
Lastly, we consider $G=SO(3)$ Chern-Simons theory at level $k=r+3$, when $r\geq5$ is an odd prime.
Larsen and Wang have showed that in this case the mapping class group of any surface $\Sigma_g$ of genus $g\geq2$ is densely represented in the projective unitary group of $\calH_{\Sigma_g}$~\cite{larsen2005density}.
We will now show that this implies that any many-torus state can be approximated to arbitrary precision:

\begin{thm}
	For $SO(3)$ Chern-Simons theory at level $k\geq8$, with $k-3$ an odd prime, we can approximate any state in the $n$-torus Hilbert space to arbitrary precision by path integration.
\end{thm}
\begin{proof}
By the Verlinde formula~\cite{verlinde1988fusion}, $\dim \HS_{\TT^2}^{\ot 2} \leq \dim \HS_{\Sigma_2}$, where $\Sigma_2$ denotes a surface of genus 2. Indeed, $\dim \HS_{\Sigma_2} = \sum_{i,j,\ell} \left( N_{ij}^\ell \right)^2$ and for each $i,j$ there is always at least one $\ell$ with $N_{ij}^\ell > 0$, so $\dim \HS_{\Sigma_2} \geq \sum_{i,j} 1 = \dim \HS_{\TT^2}^{\ot 2}$. This suggests that we can inherit the density result of~\cite{larsen2005density} by constructing an isometry from $\HS_{\TT^2}^{\ox 2}\to\HS_{\Sigma_2}$ using path integration (cf.\ \cite{barkeshli2016modular}).
We claim that the map $V$ induced by the manifold $\calV$ shown in \cref{fig:2g1tog2},~(a) is proportional to such an isometry.
Thus we need to show that $V^\dagger V \propto \id$.
Equivalently, we would like to show that the map corresponding to $-\calV \cup_{\Sigma_2} \calV$ is proportional to the identity map.

To see this, observe that, according to~\eqref{eq:M state}, $\braket{k,l|V^\dagger V|i,j}$ can be computed as an expectation value for the manifold $\overline{-\calV \cup_{\Sigma_2} \calV}$ formed by filling the inner, toroidal holes of the two copies of $\calV$ with solid tori containing Wilson loop operators in the desired representations, and then gluing along the genus-2 surface.
The very same manifold can be obtained as the connected sum of two copies of $\SS^2 \times \SS^1$.
Indeed, imagine that we first chop each $\calV$, as shown in \cref{fig:2g1tog2},~(b), then glue along the boundaries of the similarly-cut $\Sigma_2$'s, and at last along the two holes which have now turned into $\SS^2$'s.
This gives:
\begin{align*}
&\quad\braket{k,l|V^\dagger V|i,j} \\
\propto\;&Z(\overline{-\calV \cup_{\Sigma_2} \calV}; R_i,R_j,R_k^*,R_l^*) \\
=\;&Z(\SS^2 \times \SS^1 \# \SS^2 \times \SS^1; R_i,R_j,R_k^*,R_l^*) \\
\propto\;&Z(\SS^2 \times \SS^1; R_i,R_k^*) \, Z(\SS^2 \times \SS^1; R_j,R_l^*) \\
\propto\;&N_{i,0}^k N_{j,0}^l = \delta_{i,k} \delta_{j,l},
\end{align*}
where we omit the links in $Z(\dots)$ for notational simplicity;
the third line is the calculational rule~\cite[(4.1)]{witten89jones} for the connected sum $\#$,
and the fourth line follows from~\eqref{eq:fusion} together with the fact that $N_{i,0}^k = \delta_{i,k}$, since fusion with a trivial charge has no effect.
Thus we have proved that $V^\dagger V$ is proportional to an isometry.

Using this isometry, we can inherit the density result of~\cite{larsen2005density}.
Indeed, let $U$ be a unitary operator on $\HS_{\TT^2}^{\ot2}$.
Then $W = V U V^\dagger + (I - V V^\dagger)$ is a unitary on $\HS_{\Sigma_2}$.
According to~\cite{larsen2005density}, the mapping class group is represented densely in the projective unitary group of $\HS_{\Sigma_2}$.
Thus we find an orientation-preserving diffeomorphism $f$ represented by an operator proportional to some unitary $W_f$ on $\HS_{\Sigma_2}$ that approximates $W$ to arbitrary precision.
But then $U$ is approximated by $V^\dagger W_f V$ to the same precision, and the latter can be obtained by path integration over $-\calV \cup_f \calV$.
We conclude that any unitary on the two-torus Hilbert space can be approximated by path integration (up to overall rescaling).
\end{proof}

\ssection{Conclusions}%
In this work we showed how Euclidean path integrals can be used to generate interesting states like perfect tensors in simple topological field theories. The appearance of stablizer states in the $U(1)$ case allowed us to write topological formulas for the amount of tripartite entanglement in a state prepared by path integral. Furthermore, our results give a new way to encode stabilizer states in terms of topology using Chern-Simons theory. These results can clearly be extended to more general Chern-Simons theories, including theories with more $U(1)$ factors called $K$-matrix theories and more general non-Abelian theories. Our methods can also be adapted to study more general topological field theories in a variety of dimensions.

Because these field theories exhibit deconfined gauge fields, they exhibit some microscopic similarities with deconfined gauge theories having holographic duals. In particular, it will be interesting to extend our results to $SU(N)$ level $k$ Chern-Simons theories, especially in the limit of large $N$ and $k$. It will also be interesting to investigate the questions here posed in the context of free field theories, and eventually to combine the two ingredients to move towards non-trivial interacting conformal field theories, some of which like the ABJM model even have holographic duals~\cite{aharony2008N}.

\ssection{Acknowledgments}
We thank Shawn Cui, Patrick Hayden, Sepehr Nezami, \DJ{}or\dj{}e Radi\v{c}evi\'c, and Beni Yoshida for helpful discussions.
GS appreciates support from an NSERC postgraduate scholarship.
BS is supported by the Simons Foundation and Harvard University.
MW gratefully acknowledges support from the Simons Foundation and AFOSR grant FA9550-16-1-0082. 

After this work was completed, we learned that V.~Balasubramanian, J.~Fliss, R.~Leigh, and O.~Parrikar have some related results on the entanglement of multi-torus states in U(1) Chern-Simons theory.

\bibliography{references}

\end{document}